\documentclass[sn-mathphys-num]{sn-jnl}
\usepackage{multirow}%
\usepackage{amsmath,amssymb,amsfonts}%
\usepackage{amsthm}%
\usepackage{mathrsfs}%
\usepackage[title]{appendix}%
\usepackage{xcolor}%
\usepackage{textcomp}%
\usepackage{manyfoot}%
\usepackage{booktabs}%
\usepackage{algorithm}%
\usepackage{algorithmicx}%
\usepackage{algpseudocode}%
\usepackage{listings}%
\theoremstyle{thmstyleone}%
\newtheorem{theorem}{Theorem}% 
\newtheorem{prop}[theorem]{Proposition}% 
\newtheorem{lemma}{Lemma}%
\begin{document}

\title[Household Resource Allocation Dynamics and Policies: Integrating Future Earnings of Children, Fertility, Pension, Health, and Education]{Household Resource Allocation Dynamics and Policies: Integrating Future Earnings of Children, Fertility, Pension, Health, and Education}

\author[1]{\fnm{Sushmita} \sur{Kumari}}\email{sushmita.kumar@dypiu.ac.in}

\author*[2]{\fnm{Siddharth} \sur{Gavhale}}\email{siddharth.gavhale@dypiu.ac.in}
\equalcont{This author contributed equally to this work.}

\affil[1]{\orgdiv{Department of Economics}, \orgname{Vinoba Bhave University}, \orgaddress{\street{NH 33}, \city{Sindoor, Hazaribagh}, \postcode{825301}, \state{Jharkhand}, \country{India}}}

\affil*[2]{\orgdiv{School of Interdisciplinary Studies \& Research}, \orgname{DY Patil International University}, \orgaddress{\street{Sector 29 Akurdi}, \city{Pune}, \postcode{411044}, \state{Maharashtra}, \country{India}}}

\abstract{
This study presents a model that examines how families make decisions about having children, managing resources, and planning for their financial security in light of social and economic factors. It explores the balance between the number of children and the quality of life parents wish to provide, including education and future income prospects. By considering parents' expectations about their children's future earnings, the model gives new insights into how families allocate their resources. It also looks at how decisions about savings and pensions are connected to daily spending and family planning, showing the complex links between these factors. The findings offer valuable guidance for policies that address these interrelated challenges and better support families in managing their resources.}

\keywords{Household resource allocation, Utility maximization, Future wages of children, Savings and pension, Modified Quantity-Quality trade-off, Policy recommendations.}

\pacs[Acknowledgments]{The second author's research was partially supported by the seed money grant from DY Patil International University, Pune, India (Grant No. 20231028/23-24/SISR001). The authors also grateful to the anonymous referees for their valuable feedback and guidance, which improved this manuscript.}

\pacs[JEL Classification]{D13, J13 and J26}

\maketitle

%% Introduction section starts here!
\section{Introduction}Resource allocation in domestic settings is a vital factor that affects prosperity, touching on issues like human capabilities and long-lasting development of economics. With countries around the globe going through demographic changes marked by lower birth rates  in several countries and higher birth rate in some part of the world, it has become even more crucial to understand how families are adjusting resource allocation. This paper aims at generating an overall theoretical framework to analyze family decision-making processes during times of modifications in fertility patterns, particularly focusing on interactions between investments intended for children's education, saving habits, planning for retirement and the anticipated future earnings of children.

This study epitomizes an integrated approach to analyzing the various choices regarding households. Different scholars have considered fertility decisions, education investments, and savings behavior in isolation \cite{{Attanasio2003}, {Becker1960}, 
	{Hanushek1992}} , but this paper presents a holistic model that captures how they all interact with one another. In this way, it reveals how demographic transformations could influence different facets of household financial planning regarding human capital.
A tremendous pool of literature has explored the different ways people allocate resources in households. \cite{Becker1960} is one of the pioneers who started looking at how parents decide on the quantity-quality trade-off with respect to children in his seminal work regarding these issues. 
This idea was later supported by various researchers, including \cite{{Blake1981},{Hanushek1992}}, but those studies also took only one dimension of family decision making into account, thereby the whole range of interrelated choices that families are left with.

In the world of savings and retirement planning, scholars like \cite{Attanasio2003} have examined the ways in which pension systems relate to individual options for private savings. There is a possibility that public pension programs may displace private savings, emphasizing the importance of looking at every aspect of how families make financial decisions. Yet, these studies typically do not consider how such financial decisions are intertwined with fertility and human capital investment choices. Some research has started to look at the links between different dimensions of familial choice-making. For instance, a unified growth theory that ties demographic transition with human capital formation and economic development was proposed by  \cite{{GalorWeil2000}}. However, while their study is an important advance in understanding the macroeconomic dynamics of population and economic growth, it primarily centers on broad economic trends and does not extensively address the microeconomic decision-making processes within families.

The microeconomic model proposed in this paper goes further than these existing approaches with respect to its comprehensiveness by explicitly incorporating the connections between fertility, childern future wage,  education, health, savings, and pension decisions. Previous research has not looked at how anticipated future earnings of children shape present household decisions in a dynamic manner. This article's goal is to contribute to the study of a dynamic utility maximization model that takes into account fertility decisions, child education investments, health expenditures, savings, pension contributions, and the anticipated future earnings of children. With this model, we seek to clarify fertility choices and investments in human capital for children; observe the way education preference changes, health investments, and retirement savings affect the allocation of family resources; analyze how future earning expectations from children determine present consumption by their parents; as well as provide an explanation of why demographic changes might be believed to shape household economic behaviors.

This study fills a significant gap in the literature and provides useful insights for decision-makers who are concerned about the socio-economic impacts of demographic changes through a comprehensive model that combines multiple aspects of household decision-making. Careful theoretical analysis contributes to ongoing academic discourse on demographic economics and lays the groundwork for future empirical investigations into the complex interplay between factors influencing household economic behavior in light of demographic changes. The paper emphasizes on the need for a concerted policy framework that takes care of how fertility decisions interact with educational investments and retirement planning. For example, when education is subsidized through public funds, family size versus quality dilemma can be relaxed since families will not be forced by poverty to ignore their children's future earnings. Likewise, pension reforms that distinguish savings from pension contributions grant people more financial freedom so that they can spend on consumption now or save for later. This would help them balance current consumption needs and future saving plans as well as guarantee financial independence in their old age. These policies are necessary for adapting to demographic changes which induce household utility function adjustments that correspond to changes in economic conditions over time leading to sustainable economic growth and stability.

The paper is organized as follows: Section 2 outlines the theoretical framework, presenting the utility function and budget restrictions as the foundation upon which our model rests. Comprehensive analysis, including subsections that examine various aspects of household decision-making, is provided in Section 3. Section 4 addresses policy recommendations on ways to address the socio-economic consequences of changing demographics based on our results. Finally, Section 5 provides an overview of the study's contributions along with its limitations and future research directions.
%% Introduction ends here!

%%Theoretical framework start herer
\section{Theoretical Framework}
In this article, we employ a dynamic utility maximization approach to investigate the optimal allocation of household resources over a continuous time period. We consider households as two-parent families with children and presuppose that the household is a unified decision-making entity in which all individuals have similar desires and objectives. Consequently, we treat the household and the individual as analogous analytical units within a unitary household model. When this assumption is made, the analysis can proceed by modeling the behavior of households as though they were a single representative individual.

Initially, we build a utility function $u$, that represents the lifetime utility experienced by an individual born at time $t$. The utility function considers a wide range of variables that influence an individual's decision-making behavior, including consumption ($c_t$), savings ($s_t$), children's education ($e_t$), health care expenses ($p_t$), pension premiums ($q_t$), and the number of children ($n_t$) at time $t$. Moreover, the utility function accommodates a crucial aspect: the children's wage utilized by the individual, $w_{t+1}$.
The lifetime utility experienced by an individual born at time $t$ is given by the function:
\begin{eqnarray}
	u &=&  \gamma_1 \ln(c_t) + \gamma_2 \ln(n_t) + \gamma_3\ln(e_t) + \gamma_4 \ln(p_t) \nonumber \\
	&  & \hspace{2cm} + \gamma_5 \ln(n_t w_{t+1})   + \gamma_6 \ln(R_{t+1} s_t) 
	+ \gamma_7 \ln(R_{t+1}^p q_t)
	\label{eq}
\end{eqnarray}
where $\ln$ is the natural logarithm function and $R_{t+1}$, $R_{t+1}^p$ are the gross interest rates between periods $t$ and $t+1$. $s_t$ and $p_t$ are savings and pension premiums carried over from period $t$ to period $t+1$, respectively, where the composite terms $R_{t+1} s_t$ and $R_{t+1}^p q_t$ denote consumption in period $t+1$. The weights applied to the first four parameters in \eqref{eq}, i.e., $\gamma_i > 0$, for $i \in [1, 4]$ ensure their contribution to total utility is accounted for. The discount factors $\gamma_5$, $\gamma_6$, and $\gamma_7$, which fall between 0 and 1, represent the individual's preferences for different times when assessing future utilities.

To understand how people wisely allocate their household resources under changing preferences, we present a budget constraint that takes into account the trade-offs between income and costs. A household's budget reflects the fundamental trade-off between income and costs that restricts its alternatives for spending. In accordance with the standard technique in the literature \cite{{Becker1960},{BrowningLusardi1996},{Prettner2013}}, we define the budget constraint as follows:
\begin{eqnarray}
	w_t (1 - \tau n_t) = c_t + s_t + e_t n_t + p_t + q_t
	\label{eq constrain}
\end{eqnarray}
where the number of children that individuals choose to have is influenced by the fixed expense of each child, represented by $\tau$, further, it is obvious that $\tau > 0$.  The pay of an individual who gives the labor market their whole available time during period $t$ is represented by $w_t$. The individual's income is shown on the left side of \eqref{eq
	constrain}, which has been adjusted for the costs of raising children. On the other hand, the individual's expenses, which include consumption, savings, education, health care, and pension premiums, are shown on the right-hand side. This budget constraint ensures that household resources are fully allocated across different expenditure categories, reflecting the economic trade-offs inherent in fertility choices, human capital investment, and life-cycle consumption patterns.

Our idea is to analyze the utility function with respect to the budget constraint, so that we can explore how individuals optimize their choices regarding saving, investing in next period human capital, and consumption while considering the several constraints mentioned in the upcoming sections.

%%Sub section: Maximization of utility function start here
\subsection{Maximization of utility function}
Many strategies have been suggested to solve utility maximization problems in economic studies, such as Kuhn-Tucker conditions, dynamic programming, calculus-based solutions, and the Lagrangian method \cite{{Bazaraa2006},{Chiang1984},{Gibbons2007},{Stokey1989},{Varian1992}}. In this article, we use the Lagrangian approach. The Lagrangian method is particularly useful for dealing with optimization problems based on constraints. For our problem, the Lagrangian method enables the direct integration of maximization procedures with budget restrictions.
We maximize the utility function specified in (1) subject to the budget constraint, i.e., (2).
We construct the Lagrangian functional by introducing a Lagrange multiplier, $\lambda$, as follows:
\begin{equation}
	\begin{cases}
		\text{Maximize} &  \gamma_1 \ln(c_t) + \gamma_2 \ln(n_t) + \gamma_3 \ln(e_t) + \gamma_4 \ln(p_t)  + \gamma_5 \ln(n_t w_{t+1}) \\
		&  \hspace{1.25cm}  + \gamma_6 \ln(R_{t+1} s_t)  + \gamma_7 \ln(R_{t+1}^p q_t) = u\\
		\text{Subject to:} & w_t (1 - \tau n_t) -  (c_t + s_t + e_t n_t + p_t + q_t) = 0 
	\end{cases}
	\label{eq:max eq}
\end{equation}
The details of the solution to \eqref{eq:max eq} are given in Appendix \ref{appendix I}. We obtain following optimal solutions
\begin{eqnarray}
	c_t = \frac{\gamma_1 w_t}{S}, \quad
	s_t = \frac{\gamma_6 w_t}{S}, \quad
	p_t = \frac{\gamma_4 w_t}{S}, \quad
	q_t = \frac{\gamma_7 w_t}{S}, \quad
	n_t =  \frac{ \gamma_2 + \gamma_5 - \gamma_3}{\tau S} \label{sol}
\end{eqnarray}\vspace{-0.45cm}
\begin{eqnarray}
	e_t = \frac{\gamma_3 w_t\tau}{\gamma_2+\gamma_5-\gamma_3}, 
	\label{sol2}
\end{eqnarray}
where $S$ is summation of all weights and discount factors excluding $\gamma_3$, i. e., 
\begin{eqnarray}
	S = \sum_{\substack{i=1 \\ i\neq 3}}^{7} \gamma_i > 0.
	\label{sol:s}
\end{eqnarray}
In other words, the total non-educational utility weights and discount factors that people allocate to different activities such as saving, buying, having children, and paying for health care, is represented by S, which we call the \textit{utility weight sum}. 
Optimal parameters given in \eqref{sol} are inversely proportional to utility weight sum, that is, the optimal parameters in \eqref{sol} decrease as $S$ increases. It is worth to note that all those parameters are independent of educational activity, except for $n_t$. On the other hand, optimal parameter \eqref{sol2} help us to determine the ideal share for the educational spending, $\gamma_3$ with respect to $\gamma_2 + \gamma_5$.

Similar trade-offs condition without presence of $\gamma_5$ known as  'Quantity-Quality Trade-off', cf.  \cite{{beckerlewis1973},{Blake1981},{Hanushek1992}}. Thus,  with the presence of $\gamma_5$, we call it  \textit{Modified Quantity-Quality Trade-off} condition. 
In the context of this paper,  we assume that the optimal solution satisfies the \textit{Modified Quantity-Quality Trade-off} requirement, which necessitates that \( \gamma_2 + \gamma_5 \) be greater than \( \gamma_3 \).

Assuming that the optimal distribution of a family's income satisfies Modified Quantity-Quality Trade-off condition; having children and their potential earnings must be given greater weight than the expense of educating each child. 
If Parents would focus all of their efforts on training a small number of children, i.e.,   $\gamma_3 > \gamma_2 + \gamma_5$, which would result in a corner solution where $n_t = 0$. Similar results discussed in \cite{{Moav2005},{Prettner2013}} with absence of $\gamma_5$.  suggested that if the cost of child quality (which is $\gamma_3$ in our model) is sufficiently high, parents would focus on a small number of children. 
Moreover, $\gamma_3 = \gamma_2 + \gamma_5$ lead to the 
singularity situation for optimal value of $e_t$.
Thus the Modified Quantity-Quality Trade-off Condition ensures the an interior solution for maximization problem and a positive optimal number of offspring for $n_t$ and $e_t$.The intuitive notion that reproduction and child-related expenses increase in direct proportion to affluence and child choices is supported by this assumption \cite{{beckerlewis1973},{beckertomes1976}}. The proportionate values of \(\gamma_2\), \(\gamma_5\), and \(\gamma_3\) determine the optimal number of children and the amount of schooling each child receives. From an economic standpoint, the Modified Quantity-Quality Trade-off Condition yields a more relevant model by eschewing a basic corner solution \cite{{RosenzweigWolpin1980}}.
%% Section 2: Theoretical Framework end here!

%%
%% Section 3 Findings and Analysis Start here!
\section{Findings and Analysis}
This section presents and analyzes the key findings derived from our theoretical model, organized into three subsections based on thematic coherence and the nature of the insights provided. The results collectively provide a comprehensive theoretical framework for understanding how parameters included in the model influence various household socioeconomic decisions. The results offer insights into the complex interplay between fertility choices, education investments, savings behavior, and pension planning.
%%Sebsection 3.1 Start here!!
\subsection{Impact of Utility Weights on Household Resource Allocation}
Our analysis begins by examining the fundamental relationships between utility weights and optimal household choices. This first set of results elucidates how variations in preference parameters influence resource allocation across different expenditure categories, with particular emphasis on the distinctive role of education in household decision-making. Analyzing how partial derivatives of optimal decision variables react to shifts in their corresponding utility weight reveals the mechanism by which households adjust to new aspects. These findings do not only help improve comprehension of the choices people make concerning their children but can also elaborate more about preference changes impacts on family financial management.
%%Lemma 1 :: 1st Result
\begin{lemma}[Optimal Choice Changes with Allocation Weights]
	In a framework of multiple choice decision-making involving consumption, savings, health spending, pension premiums, number of children, and their education, the optimal allocation of resources to each choice is directly influenced by the weights assigned to them.
	\label{lem: 1}
\end{lemma}
\begin{proof}
	Appendix \ref{Appendix: lemma 1 solution}
\end{proof}
The results indicate that increasing the weight \(\gamma_i\), where $i \in (1,7)\setminus \{5\}$, on any particular choice directly increases the utility derived from that choice. To optimize utility subject to the budget constraint, individuals will allocate more resources towards the choice with the higher weight. For instance, a higher \(\gamma_1\) increases the utility from consumption \(c_t\), leading to a higher optimal \(c_t\). This principle holds similarly for \(s_t\), \(p_t\), \(q_t\), \(n_t\), and \(e_t\). The weights \(\gamma_i\) thus play a critical role in determining how resources are distributed among different uses.
The economic implication is that individuals prioritize their spending and saving decisions based on the relative utility weights assigned to different expenditures. Changes in these weights shift the optimal allocation, reflecting the trade-offs individuals are willing to make to maximize their overall utility. This aligns with standard economic theory \cite{{MasColell1995},{Varian2014}}, which postulate that individuals optimize their utility given budget constraints and preferences.
Our next result gives the crucial insights into the unique role of education in resource allocation.
%%Lemma 2 :: 2nd Result
\begin{lemma}[Independence of Optimal Choices from Education Weight]
	In the context of household resource allocation, the weight assigned to child education (\(\gamma_3\)) does not affect the optimal levels of consumption (\(c_t\)), savings (\(s_t\)), health spending (\(p_t\)), and 
	pension premiums (\(q_t\)). This indicates that the decision to invest in child education operates almost independently of other financial decisions within the household.
	\label{lem: 2}
\end{lemma}
\begin{proof}
	Appendix \ref{Appendix: lemma 2 solution}
\end{proof}
The zero value of derivatives demonstrates that variations in $\gamma_3$ have no impact on the optimal levels of consumption, savings, health spending, and pension premiums.
This lemma brings out that educational investments can be separated from other financial decisions made by a household; only the value of $\gamma_3$ is significant in determining the total number of children ($n_t$) as well as expenses incurred in educating them, the two being mutually exclusive but having no relationship with any other form of spending money. These findings suggest that educational decisions may be driven by factors distinct from those governing other forms of consumption and savings, which is a reason why a specialized focus on human capital investment is supported by economic models of household behavior.

%% Additional results after lemma 2:
Furthermore, we calculate following derivatives 
\begin{eqnarray}
	\frac{\partial e_t}{\partial \gamma_1} = \frac{\partial e_t}{\partial \gamma_4} =  \frac{\partial e_t}{\partial \gamma_6} = \frac{\partial e_t}{\partial \gamma_7} = 0.
	\label{eq: results e_t}
\end{eqnarray}
Independence of education expenditure $e_t$ from the utilitarian weights on consumption, health expenditures, savings, and pension premiums is further confirmed. This finding emphasizes how unique the schooling decision variable $e_t$ is in the model as it is only influenced by the weights $\gamma_2, \gamma_3$ and $\gamma_5$, which are associated with the number of children and their potential earnings in the future. This division highlights the special function of education in the model even more, as it influences the financial restriction by influencing the number of children, while the other choices are driven by their respective utility weights and the trade-offs between non-education activities.

Together, Lemma \ref{lem: 1}, Lemma \ref{lem: 2}, and \eqref{eq: results e_t} shows how different household resource allocation becomes when preferences change. These results propose that a number of determinations that especially concern schooling are guided by another set of motives although households generally react logically to their particular taste formations.
%%Subsection 3.1 Ends here!!

%%Sub Section 3.2 starts here: _
\subsection{Interactions Between Pension Preferences, Savings Behavior, \& Family Planning}
This subsection examines how decisions about consumption, health care costs, and fertility made by households are impacted by shifts in pension savings preferences. It looks at how  savings and pension premiums are traded off, as well as how changing priorities impact other important decisions. The results provide a theoretical framework that highlights the need to strike a balance between investments in the next generation, future financial stability, and present consumption. 
%%Lemma 3 :: 3rd Result
\begin{lemma}(Effects of Consumption, Health, and Savings Preferences on Fertility). The partial derivatives of the optimal number of children (\(n_t\)) with respect to the utility weights on consumption (\(\gamma_1\)), health expenditure (\(\gamma_4\)), and savings (\(\gamma_6\)) are equal and negative:
	\label{lem: 3}
\end{lemma}
\begin{proof}
	Appendix \ref{Appendix: lemma 3 solution}
\end{proof}

This lemma captures the complex relationships between fertility decisions and preferences for consumption, health expenditure, and savings in the context of household resource allocation. This result indicates that as households increase their preference for consumption, health expenditures, or savings , they tend to allocate fewer resources towards raising additional children. The negative and equal partial derivatives suggest a uniform trade-off between these preferences and fertility, emphasizing the impact of household budget constraints on family size decisions

Several factors influence these relationships:
\begin{itemize}
	\item \(\tau\): The fixed cost per child, which amplifies the negative effects of \(\gamma_1\), \(\gamma_4\), and \(\gamma_6\) on fertility. Higher costs per child reduce the affordability of additional children, making households more selective about the number they have.
	\item \(\gamma_2\) and \(\gamma_5\): The utility weights on the number of children and their future earnings. Stronger preferences in these areas mitigate the negative effects of \(\gamma_1\), \(\gamma_4\) and \(\gamma_6\), as families that place high value on children's future earnings may be more inclined to have more children despite other costs.
	\item \(\gamma_3\): The utility weight on child education. A higher \(\gamma_3\) enhances the negative effects of \(\gamma_1\), \(\gamma_4\), and \(\gamma_6\) on fertility by increasing the perceived cost of raising children due to educational expenses.
\end{itemize}

This result enriches our understanding of fertility decisions by highlighting their inter-connectedness  with other aspects of household decision-making. It underscores the complexity of demographic changes and provides a theoretical foundation for exploring how evolving societal preferences might influence fertility trends. This interplay is crucial for policy-makers aiming to address demographic shifts and design interventions that balance family well-being with economic stability.
Building on our understanding of fertility decisions, we now turn our attention to how pension preferences influence various household choices.
%%Lemma 4 :: 4th Result
\begin{lemma}(Effects of Pension Premium Weight on Household Decisions). The partial derivatives of optimal consumption ($c_t$) and health expenditure ($p_t$) with respect to the utility weight on pension premiums ($\gamma_7$) are unconditionally negative. On the other hand,  partial derivatives of optimal number of children ($n_t$) with respect to the utility weight on pension premiums is negative, provided that \textit{Modified Quantity-Quality Trade-off} condition satisfies.
	\label{lem: 4}
\end{lemma}
\begin{proof}
	Appendix \ref{Appendix: lemma 4 solution}
\end{proof}
This result captures the complex interplay between pension premium related decisions and other household choices. The negative relationships arise because an increase in the utility weight placed on pension premiums ($\gamma_7$) leads to a reallocation of resources towards pension contributions, necessitating reductions in other areas of household expenditure.
The negative impact on consumption and health expenditure suggests that households face a trade-off between current consumption and saving for retirement. As the utility weight on pension premiums increases, households are compelled to cut back on their current spending and health investments to allocate more towards pension contributions. This adjustment reflects a substitution effect where future-oriented savings (pension premiums) replace immediate consumption and health-related expenditures.

The reduction in the number of children due to increased pension premiums highlights the broader economic implications. Higher pension contributions may limit available resources for raising children, leading to a decrease in fertility rates. This result underscores the interplay between retirement planning and family planning, illustrating how prioritizing retirement savings can have significant effects on family size and current living standards.
Having examined the effects of pension preferences on consumption, health expenditure, and fertility, we now focus on the critical relationship between pension premiums and  savings.
%%Theorem 5 :: 5th Result
\begin{theorem}
	[Optimal Savings Decrease with Higher Weight on Pension Premiums]
	An increase in the weight assigned to pension premiums (\(\gamma_7\)) in the household decision-making framework leads to a decrease in the optimal level of savings (\(s_t\)).
	\label{lem: 5}
\end{theorem}
\begin{proof}
	Appendix \ref{Appendix: lemma 5 solution}
\end{proof}
When the utility weight on pension premiums (\(\gamma_7\)) increases, households are incentives to allocate a greater portion of their income towards pension contributions, consequently reducing the amount allocated to  savings (\(s_t\)).
From an economic perspective, this inverse relationship underscores the substitution effect between savings and pension premiums. Both savings and pension premiums are methods for ensuring future financial security. However, when households place a higher utility weight on pension premiums, they prioritize contributions to pension funds over  savings.

Key economic insights from this relationship include, that,  higher wages (\(w_t\)) result in a larger decrease in savings when \(\gamma_7\) increases. 
This negative relationship arises because savings ($s_t$) and pension premiums ($q_t$) are substitutes for transferring resources across time periods to support future consumption. This result emphasizes how shifts in pension preferences influence not only the amount of current savings but also the overall financial planning of households. As more income is diverted towards pensions, less is available for other savings goals, impacting long-term financial security and consumption patterns.
As individuals are incentivized to pay more into pensions, they optimally reduce their private savings levels, with the magnitude depending on preferences and income levels. The lemma quantifies this trade-off, providing an insightful component of understanding optimal life-cycle behavior across savings and retirement decisions.

Overall, these findings provide a comprehensive understanding of how household decisions are intertwined with pension savings preferences. The results reflect the broader economic implications of shifting priorities, emphasizing the need for a balanced approach to retirement planning and family investment. The interplay between pension contributions, savings behavior, and family planning offers valuable insights for designing policies that address the challenges posed by demographic changes and aging populations.
%%Subsection 3.2 Ends here!!

%%Sub Section 3.3 starts here: 
\subsection{Education-Fertility Trade-off and Future Earnings Impact: Household Responses to Demographic Change}
This subsection illustrates the relationship between education and fertility, with children's potential future incomes exerting an influence on households' decisions. Further, it looks at how households respond to shifting demographic pressures and shifting preferences by adjusting how they allocate their resources. It looks at how future wages for children affect savings, consumption, and fertility decisions.  
The results have significant implications for family planning programs, education policies, and long-term economic growth strategies in countries undergoing demographic transitions.
%%Lemma 6 :: 6th Result
\begin{prop}[Optimal Number of Children Decreases with Higher Weight on Child Education]
	An increase in the utility weight on child education ($\gamma_3$) leads to a decrease in the optimal number of children ($n_t$).
	\label{lem: 6}
\end{prop}
\begin{proof}
	Appendix \ref{Appendix: lemma 6 solution}
\end{proof}
This result highlights a critical aspect of household economic behavior: the trade-off between the number of children and the investment in each child's education. Economically, this inverse relationship is driven by the limited resources available to households. As parents place greater importance on education (higher $\gamma_3$), they must allocate more of their budget towards educational expenses, which reduces the funds available for raising additional children. This decision is influenced by:

\textit{Fixed Cost of Children ($\tau$)}: Higher fixed costs per child mean that as education becomes a higher priority, the financial burden of maintaining each child increases, leading to a preference for fewer children.

\textit{Resource Allocation}: Households with constrained budgets need to make trade-offs between the quantity and quality of children. As the utility weight on education increases, households allocate more resources to fewer children to enhance their future earning potential.

This result is crucial for understanding how changes in parental preferences and economic conditions can impact fertility decisions. It shows that policy measures aimed at increasing educational investments may inadvertently lead to lower fertility rates, as parents balance the quality of education against the quantity of children. The result emphasizes the importance of considering both educational policies and family planning strategies in tandem to achieve desired demographic and educational outcomes.
The negative sign indicates that as the utility weight on child education ($\gamma_3$) increases, the optimal number of children ($n_t$) decreases. This is intuitive because a higher weight on education implies a preference for fewer children but higher investment in each child's education.
The magnitude of the decrease in $n_t$ is inversely proportional to the fixed cost of children ($\tau$) and the sum of non-education utility weights ($S$). This is possible because a higher fixed cost or higher non-education utility weights would lead to a smaller trade-off between child quantity and quality.
Complementing our previous finding, we now investigate the reciprocal relationship between the preference for larger families and educational investment.
%%Proposition 7 :: 7th Result
\begin{prop}[Optimal Education Expenditure Decreases with Higher Weight on Number of Children]
	An increase in the utility weight assigned to the number of children ($\gamma_2$) leads to a decrease in the optimal education expenditure per child ($e_t$).
	\label{lem: 7}
\end{prop}
\begin{proof}
	Appendix \ref{Appendix: lemma 7 solution}
\end{proof}
This result articulates a fundamental economic trade-off between the quantity and quality of children within household decision-making. As the utility weight on the number of children ($\gamma_2$) increases, parents are incentivized to have more children. Given a fixed household budget, this increased preference for a higher number of children leads to a reduction in the resources allocated per child, specifically for education expenditure ($e_t$).
Key economic insights from this relationship include:
\begin{itemize}
	\item  \textit{Income Effect}: Higher household income generally leads to higher overall education expenditure. However, as $\gamma_2$ increases, the decrease in $e_t$ becomes more pronounced due to the larger initial allocation to education.
	\item \textit{Preference Impact}: A higher utility weight on education ($\gamma_3$) implies more resources are initially allocated to education expenditure. Therefore, the reduction in $e_t$ is amplified when $\gamma_2$ increases.
	\item  \textit{Fixed Costs}: Higher fixed costs per child ($\tau$) reduce the resources available for both the quantity and quality of children, influencing the magnitude of the decrease in $e_t$.
	\item  \textit{Quantity-Quality Trade-off}: The modified trade-off condition reflects how increased preference for child quantity impacts the allocation of educational resources.
\end{itemize}
This result demonstrates how shifts in preferences towards having more children necessitate adjustments in educational investments, providing a clear quantitative measure of this trade-off. It is critical for understanding optimal fertility and human capital investment strategies within the model.
Expanding our analysis to consider the long-term perspective, we now examine how the perceived importance of children's future earnings influences current household decisions.
%%Thoerem 8 :: 8th Result
\begin{theorem}[Impact of Children's Future Earnings Weight on Household Decisions]
	An increase in the weight assigned to children's future earnings  (\(\gamma_5\)) in the household decision-making framework leads to a decrease in the optimal level of current consumption (\(c_t\)), savings (\(s_t\)), health spending (\(p_t\)), or pension premiums (\(q_t\)), while also leading to an increase in the number of children (\(n_t\)).
	\label{lem: 8}
\end{theorem}
\begin{proof}
	Appendix \ref{Appendix: lemma 8 solution}
\end{proof}
As the weight \(\gamma_5\) on the utility derived from children's future earnings increases, households are incentivized to reallocate resources away from current consumption (\(c_t\)), savings (\(s_t\)), health spending (\(p_t\)), or pension premiums (\(q_t\)), towards increasing the quantity (\(n_t\)) and quality (\(e_t\)) of children. This reallocation is necessitated by the budget constraint, as higher investments in children's education and future earnings reduce the available resources for present consumption, savings,  health spending or pension premiums.

The negative sign of the partial derivatives \(\frac{\partial c_t}{\partial \gamma_5}\),  \(\frac{\partial s_t}{\partial \gamma_5}\),  \(\frac{\partial p_t}{\partial \gamma_5}\),  \(\frac{\partial q_t}{\partial \gamma_5}\) signifies that emphasizing children's future earnings in the household's utility function leads to a decline in optimal current consumption, savings, health spending or pension premiums. This result aligns with the notion that households are willing to sacrifice resources to enhance their children's human capital and future earnings potential when these factors are valued more highly. Whereas the  partial derivative of the optimal number of children (\(n_t\)) with respect to \(\gamma_5\) is positive which shows the  number of children (\(n_t\)) increases, suggesting that higher valuation of children's future earnings incentivizes having more children.

This lemma highlights the critical tradeoffs faced by households in resource allocation, demonstrating how increasing the importance placed on children's future earnings (\(\gamma_5\)) influences the optimal choice of  other resources.
These results collectively demonstrate that an increased emphasis on children's future earnings leads to a significant reallocation of household resources. Families reduce current consumption and other expenditures to invest more in both the quantity and quality of children, anticipating higher returns from their children's future earnings. This finding underscores the complex interplay between fertility decisions, human capital investment, and intertemporal resource allocation in household decision-making.
%%Subsection 3.3 Ends here!!

%%Sub Section 3.4 starts here: 
\subsection{Comprehensive Overview of Partial Derivatives of Optimal Solutions}
In order to present a complete view of what the model implies, we provide a detailed table of first order partial derivatives. This analysis summarizes in brief how changes in the rates of utility (\(\gamma_i\)) lead to variations in different household decision parameters. For the notation simplification we use $\frac{w_t}{S^2}  = \beta.$ We tabulate all the results in Table \ref{tab:partial_derivatives}. 
\renewcommand{\arraystretch}{2.14} % Increase row height
\begin{table}[h]
		\caption{Partial Derivatives of Decision Parameters with Respect to $\gamma_i$}
	\centering
	\begin{tabular}{c||c|c|c|c|c|c|c|}
		& $c_t$ & $n_t$ & $e_t$& $p_t$ & $s_t$ & $q_t$ \\ \hline 
		$\gamma_1$ & $\frac{S-(\gamma_1+\gamma_3)}{\beta^{-1}}$& $-\frac{(\gamma_2 + \gamma_5 - \gamma_3 )}{\tau S^2}$ & 0 & $- \beta\gamma_4$  & $-\beta\gamma_6$& $- \beta\gamma_7$ \\
		$\gamma_2$ & $-\beta \gamma_1$ & $\frac{ S-(\gamma_2+\gamma_5)}{\tau S^2}$ & $-\frac{w_t \gamma_3 \tau}{(\gamma_2 + \gamma_5 - \gamma_3)^2}$& $-\beta \gamma_4$ & $-\beta \gamma_6$ & $-\beta\gamma_7$ \\
		$\gamma_3$ & 0 & $-\frac{1}{\tau S}$ & $\frac{w_t \tau(1+\gamma_3)}{(\gamma_2 + \gamma_5 - \gamma_3)^2}$ & 0 & 0 & 0 \\
		$\gamma_4$ & $-\beta \gamma_1$& $-\frac{(\gamma_2 + \gamma_5 - \gamma_3)}{\tau S^2}$ & 0 & $\frac{S-(\gamma_4+\gamma_3)}{\beta^{-1}}$ & $-\beta \gamma_6$& $-\beta \gamma_7$ \\
		$\gamma_5$ & $-\beta \gamma_1$ & $\frac{ S-(\gamma_2+\gamma_5)}{\tau S^2}$ & $-\frac{w_t \gamma_3 \tau}{(\gamma_2 + \gamma_5 - \gamma_3)^2}$& $-\beta \gamma_4$ & $-\beta \gamma_6$ & $-\beta \gamma_7$ \\
		$\gamma_6$ & $-\beta \gamma_1$& $-\frac{(\gamma_2 + \gamma_5 - \gamma_3)}{\tau S^2}$ & 0& $-\beta\gamma_4$  & $\frac {S-(\gamma_6+\gamma_3)}{\beta^{-1}}$ & $-\beta \gamma_7$ \\
		$\gamma_7$ & $-\beta \gamma_1$ & $-\frac{(\gamma_2 + \gamma_5 - \gamma_3)}{\tau S^2}$ & 0& $-\beta \gamma_4$ & $-\beta \gamma_6$& $\frac{S-(\gamma_7+\gamma_3)}{\beta^{-1}}$ 
	\end{tabular}
	\label{tab:partial_derivatives}
\end{table}

Some key additional insights include:
\begin{enumerate}
	\item \textbf{Symmetry in impacts}: There is a distinct resemblance in how alterations in utility weights alter diverse decision parameters. For example, an increase in consumption ($\gamma_1$) negatively affects savings ($s_t$), health spending ($p_t$), and pension premiums ($q_t$) equally.
	
	\item \textbf{Unique impact of educational weight}: Educational weight ($\gamma_3$) is the only parameter that does not directly influence consumption, savings, health expenditures, or pension premiums. This supports earlier findings regarding education decisions being relatively autonomous.
	
	\item \textbf{Cross-impacts}: Cross-impacts between different utility weights and decision parameters are displayed in the table. For instance, a rise in savings weight ($\gamma_6$) has a negative consequence on consumption, health expenditure, and number of children, but on the contrary, has positive effects on the pension premiums.
	
	\item \textbf{Magnitude of effects}: Comparative analysis of effect sizes across various parameters provides insight into which utility weights have the greatest impact on particular choices.
\end{enumerate}
In our study, the tabular structure is essential since it simplifies comprehension, unifies all mathematical equations, and makes it easier to compare utility weights and decision factors. It facilitates the identification of patterns and symmetries in a model and acts as a resource for scholars and decision-makers to comprehend how preferences influence household decisions.
%% Section 3: Findings and Analysis End here!

%% Section 4 Policies Start here!
\section{Policy Implications and Recommendations}
In this paper, we have identified several key findings that  include the trade-off between the number and quality of children, changes in private savings and pension contributions, and the influence of early childhood investments on long-term outcomes. Based on these results, we propose a set of policies to address the socio-economic impacts. These policy recommendations are grounded in both our findings and existing literature.

\subsection{Enhancing Education Systems and Human Capital}
The observed trade-off between the number of offspring and their well-being, consistent with \cite{{Becker1960}} postulate, suggests that improving human capital through education is crucial. To address this, we recommend the following measures
\begin{itemize}
	\item Educate and increase subsidized education programs, such as free primary and secondary education for all children, and scholarships for higher education. Consequently, people can study regardless of their parents' financial status. In India, for instance, the Right of Children to Free and Compulsory Education (RTE) Act was introduced in 2009, but implementation and reach are still weak \cite{{ministry2020}}. On the other hand, these higher education schemes like Prime Minister's Special Scholarship Scheme (PMSSS) should have more students and disciplines involved \cite{{ministry2021a}}.
	
	\item Deliberate on early childhood program initiatives aimed at stimulating cognitive, social, and affective growth. Indeed, the National Education Policy 2020 of India is an indication for substantial commitment to early childhood care and education (ECCE) through its insistence on ECCE being part of the formal educational system \cite{{ministry2020}}. Research from the international arena supports this proposition as it reveals that there are considerable advantages of starting these measures before a child enters school \cite{{Heckman2006},{OECD2017}}.
	\item Enhancing Science, Technology, Engineering, and Mathematics (STEM) and vocational training to prepare students for future job markets. The Skill India Mission being promoted by the Government of India is to train more than four hundred million individuals in several different skills before the year 2022. This initiative should be improved and aligned with global best practices that address the changing needs of labor markets \cite{{ILO2019},{WEF2020}}.
\end{itemize}

\subsection{Strengthening Financial Support for Families}
The financial constraints and family planning decisions interact in a complex way. Therefore, we propose some financial incentives in order to ease families economic burden and motivate spending on children's growth. These recommendations seek to establish an enabling environment for families who are making qualitative-quantitative choices regarding children.
\begin{itemize}
	\item A way to lessen the weight of parents is by coming up with more education and child tax credits. Although in India there are schemes such as Sukanya Samriddhi Yojana which is meant to promote saving for the education of girl children, it would also be possible to have comprehensive education expense tax credits \cite{{ministry2021a},{UNESCO2020},{WorldBank2018}}.
	
	\item Encourage employers to contribute to savings plans and promote flexible working conditions. India's New Labour Codes, especially the Code on Social Security, 2020 gives a structure for improving social security benefits \cite{{MoLE2020}}. These should be effectively implemented and extended to include flexible working arrangements as per global best practices \cite{{Eurofound2017},{ILO2018}}.
\end{itemize}

\subsection{Expanding Access to Affordable Healthcare}
Access to quality healthcare becomes an important aspect that influences family decision-making especially with the decline in birth rates. Therefore, this part suggests several policies aimed at improving accessibility and affordability of healthcare related to preventive care as well as total maternal and child health service delivery.
\begin{itemize}
	\item Increase endeavors to achieve universal access to health care and put more resources into preventive health initiatives. The Ayushman Bharat program in India is targeting 500 million beneficiaries for health insurance coverage \cite{{ministry2021c}}. As per WHO guidelines on universal health coverage \cite{{MorenoSerra2012},{WHO2019}}, there should be ongoing evaluations and expansion efforts under this scheme so that it can cater for even more people in the population.
	\item Bolster an all-encompassing approach to maternal and child health service delivery, including food and nutrition initiatives. Although there have been significant improvements with regard to this in India's National Heath Mission, malnutrition is still an area requiring attention \cite{{MoWCD2021}}. In addition to drawing from global models \cite{{Black2017},{UNICEF2019}}, such actions need to be ramped up.
\end{itemize}

\subsection{Securing Retirement and Reforming Pension Systems}
Our research reveals the complex connections that exist among pension choices, individual financial reserves as well as decisions about having children. Based on these revelations, we suggest a number of changes in pensions and savings that can improve retirement security in accordance with changing demographics and household expenditure patterns.
\begin{itemize}
	\item Choose plans with variable rates or charge enrollees a fixed fee. This will guarantee good returns through different investment options in every market situation. The National Pension System (NPS) is the first rung but needs more coverage and efficiency in its implementation \cite{{PFRDA2021}}. In designing and putting into practice various pensions systems, there are potential lessons that could come from several countries' experiences \cite{{beshears2009},{WorldBank2019}}.
	\item Private savings should have more incentives while comprehensive financial literacy programs should be in place. These acts should be magnified and diversified on the backdrop of worldwide proof regarding the significance of financial education \cite{{Lusardi2014},{OECDINFE2020}}.
\end{itemize}

\subsection{Labor Market Policies}
The results of this research have far-reaching consequences for labor markets as it accounts for fertility changes and the investment of human capital in relation to changing workforce requirements. This part presents some policy suggestions aimed at matching educational achievements with anticipated demand for jobs while at the same time ensuring better work-life balance through the greater economic implications created by demographic transitions.
\begin{itemize}
	\item Make sure that educational plans as well as vocational training initiatives fit anticipated future job demands. The main objective of the Indian National Skill Development Mission is to bring together skill-building efforts from all corners of the country. This should be regularly modified depending on job forecasts and worldwide developments \cite{{McKinsey2018},{ILO2019}}.
	\item Policies that support work-life balance should be put in place and reinforced, among them parental leave and flexible working conditions. India has made maternity leave compulsory but there is room for paternity leave and more flexible working circumstances. Drawing from global best practices may help these policies to assist people in balancing their careers with family obligations \cite{{OECD2016},{Thevenon2013}}.
\end{itemize}
The policy recommendations outlined above address the key socio-economic impacts identified in this paper. By promoting education, providing financial incentives, enhancing healthcare access, and strengthening retirement security, policymakers can create a supportive environment that encourages families to invest in their children's development and ensure long-term economic stability. These policies not only address the immediate challenges posed by falling fertility rates but also  contributing to sustainable economic development and improved quality of life for future generations.
%% Section 4: Policies End here!

%% Section 5 Discussion Start here!
\section{Discussion}
This research explored the intricate relationship between reproductive choices, educational spending, savings habits, and retirement planning in the setting of dwindling birth rates. Our theoretical framework offers valuable perspectives that help us understand how families organize their resources in response to fluctuating population patterns.
Findings obtained by us reveal the significant trade-offs between child quantity and quality that were consistent with \cite{{Becker1960}} seminal work. Hence, parents would ideally have fewer children when focusing on educational matters (having higher values for $\gamma_3$), but spend more resources in educating their offspring. This finding can be related to previous studies done by \cite{{Blake1981},{Hanushek1992}} which give a numerical perspective of how changes in parental preferences affect fertility rates and human capital investments.

Interestingly, our model suggests that the decisions made about education expenditure are more or less autonomous from other decisions regarding household finance. Such a separation infers a different set of factors that motivate educational investments compared to those affecting consumption and hence savings, thus justifying the specialized emphasis on human capital investment within economic models of household behavior. Our study also discovered a complicated interplay between preferences for pension schemes and personal savings. Since families assign more importance to pension premiums, their levels of private saving and present-day consumption will drop. This trade-off in pension contributions as opposed to private saving expands earlier research by \cite{{Attanasio2003}}, on pension asset accumulation as well as home savings. Our results suggest that policy changes affecting pension systems may have unintended consequences on private savings behavior.

In an unexpected way, our model suggests that by giving priority to future earnings of children, higher fertility rates would be realized. This contradicts some current theories that predict declining fertility due to parents investing more in their child's human capital \cite{{GalorWeil2000}}. A possible reason for this result is that parents see having many children and thus a large family as a sort of investment in expected family income, especially where inter-generational transfers are important.
It is important to acknowledge the limitations of our study. Our theoretical layout is comprehensive to some extent; however, it is based on certain simplifications. Hence decisions made in the real world are probably more complicated than those in our model, thus leaving out such things as cultural patterns, policy settings, and economic uncertainties. Moreover, our model posits a single household decision-making procedure characterized by it may not be able to reflect fully intra-household bargaining dynamics.
However, these constraints have important ramifications for government action and future studies. The close relationship between educational choices and childbearing choices implies that measures intended to raise the standard of education may inadvertently affect birth rates. As such, these possible compromises should be borne in mind by policymakers while establishing programs relating to both schooling as well as controlling family growth. Moreover, our findings emphasize that it is necessary to adopt a unified approach towards planning for retirement as well as family assistance systems. Since pension choices affect both monetary saving actions and reproductive choices, any policy affecting pension schemes can produce outspill financial strategies.

Our model might be extended in future research by including additional practical assumptions that relate to the traditional setting of household decision-making processes, such as intra-household negotiation and ambiguity. In particular, it would be useful to test our model predictions in various cultural and economic settings through empirical studies. In addition, understanding how these changes in technology and job markets affect the interconnections we discovered would inform policymakers' long-term strategies.

To sum up, this research offers an all-encompassing theoretical structure for making sense of household resource distribution when fertility preferences are changing. In shedding light on the intricate linkages between education finance, saving choices, pension schemes, and procreation choices, we add to a deeper comprehension of demographic shifts and their economic effects. Such knowledge can aid in developing better-rounded policies that respond adequately to declining birth rates in numerous nations.
%% Section 5: Discussion End here!

%% Disclosure and Acknowledgments starts here!!
\section*{Declarations}
The authors have no competing interests to declare that are relevant to the content of this article. Both the authors agreed to publish this article.

%%Appendix starts here!!
\begin{appendices}
	\section{Solution of Maximization problem}
	\label{appendix I}
	Construct the Lagrangian functional by introducing a Lagrange multiplier $\lambda$, as given below
	\begin{align*}
		L &= \gamma_1 \ln(c_t) + \gamma_2 \ln(n_t) + \gamma_3 \ln(e_t) + \gamma_4 \ln(p_t) 
		+ \gamma_5 \ln(n_t w_{t+1}) + \gamma_6 \ln(R_{t+1} s_t) + \gamma_7 \ln(R_{t+1}^p q_t) \\
		&\quad + \lambda \left( w_t (1 - \tau n_t) -  (c_t + s_t + e_t n_t + p_t + q_t) \right)
	\end{align*}
	Subsequently, we compute partial derivatives of the Lagrangian with respect to each decision variable and the Lagrange multiplier and set them equal to zero, i.e., first order condition:
	\[
	\begin{aligned}
		\frac{\partial L}{\partial c_t} &= \frac{\gamma_1}{c_t} - \lambda = 0 \\
		\frac{\partial L}{\partial n_t} &= \frac{\gamma_2 + \gamma_5}{n_t} - \lambda (\tau w_t + e_t) = 0 \\
		\frac{\partial L}{\partial e_t} &= \frac{\gamma_3}{e_t} - \lambda  n_t = 0 \\
		\frac{\partial L}{\partial p_t} &= \frac{\gamma_4}{p_t} - \lambda = 0 \\
		\frac{\partial L}{\partial s_t} &= \frac{\gamma_6}{s_t} - \lambda = 0 \\
		\frac{\partial L}{\partial q_t} &= \frac{\gamma_7}{q_t} - \lambda = 0 \\
		% \frac{\partial L}{\partial \tau} &=  w_t n_t \lambda = 0 \implies \lambda = 0 \\
		\frac{\partial L}{\partial \lambda} &= w_t (1 - \tau n_t) -  (c_t + s_t + e_t n_t + p_t + q_t) = 0
	\end{aligned}
	\]
	By eliminating $\lambda$ from above, we obtain
	\[
	\frac{\gamma_1}{c_t} = \frac{\gamma_2 + \gamma_5}{n_t (\tau w_t + e_t)} = \frac{\gamma_3}{e_t  n_t} = \frac{\gamma_4}{p_t} = \frac{\gamma_6}{s_t} = \frac{\gamma_7}{q_t}
	\]
	
	Solving these equations simultaneously yield the optimal values as follow
	\begin{align*}
		c_t = \frac{\gamma_1 w_t}{S},  \quad
		s_t = \frac{\gamma_6 w_t}{S}, \quad
		p_t = \frac{\gamma_4 w_t}{S}, \quad
		q_t = \frac{\gamma_7 w_t}{S}, \quad  
		e_t = \frac{\gamma_3 w_t\tau}{\gamma_2+\gamma_5-\gamma_3}, \quad
		n_t =  \frac{ \gamma_2 + \gamma_5 - \gamma_3}{\tau S},   
	\end{align*}
	where $S = \sum_{i=1, i\neq 3}^{7} \gamma_i$. Note that $S$ is non-negative number. 

	\section{Proof of lemmas, propositions and theorems}
	\subsection{Proof of Lemma \ref{lem: 1}}
	\label{Appendix: lemma 1 solution}
	\begin{proof}
		Take the partial derivative of optimal solutions, i.e., \eqref{sol} and \eqref{sol2}  with respect to its corresponding weight \(\gamma_i\), where $i \in (1,7)\setminus \{5\}$, mentioned in \eqref{eq:max eq}. That leads to
		\begin{eqnarray*}
			\frac{\partial c_t}{\partial \gamma_1} &=& \frac{w_t (\gamma_2 + \gamma_4 + \gamma_5 + \gamma_6 + \gamma_7)}{S^2} > 0, \qquad
			\frac{\partial p_t}{\partial \gamma_4} \quad = \quad \frac{w_t (\gamma_1 + \gamma_2 + \gamma_5 + \gamma_6 + \gamma_7)}{S^2} > 0,\\
			\frac{\partial s_t}{\partial \gamma_6} &=& \frac{w_t (\gamma_1 + \gamma_2 + \gamma_4 + \gamma_5 + \gamma_7)}{S^2} > 0,
			\qquad
			\frac{\partial q_t}{\partial \gamma_7} \quad = \quad \frac{w_t (\gamma_1 + \gamma_2 + \gamma_4 + \gamma_5 + \gamma_6)}{S^2} > 0.
			\\
			\frac{\partial n_t}{\partial \gamma_2} &=& \frac{\gamma_1 + \gamma_3 + \gamma_4 + \gamma_6 + \gamma_7}{\tau (S)^2} > 0, \qquad \hspace{0.65cm}
			\frac{\partial e_t}{\partial \gamma_3} \quad = \quad \frac{w_t \tau (1 + \gamma_3)}{(\gamma_2 + \gamma_5 - \gamma_3)^2} > 0.
		\end{eqnarray*}
		The results are direct from $\gamma_i > 0$, $\tau >0$ and $w_t > 0$. 
	\end{proof}

	\subsection{Proof of Lemma \ref{lem: 2}}
	\label{Appendix: lemma 2 solution}
	\begin{proof}
		Take the derivatives of optimal solutions provided in \eqref{sol} with respect to \(\gamma_3\) to obtain the following
		\[
		\frac{\partial c_t}{\partial \gamma_3} = \frac{\partial s_t}{\partial \gamma_3} = \frac{\partial p_t}{\partial \gamma_3} = \frac{\partial q_t}{\partial \gamma_3} = 0.
		\]
		Since \(\gamma_3\) does not appear in the expressions for the optimal \(c_t\), \(s_t\), \(p_t\), and \(q_t\), variations in \(\gamma_3\) have no impact on these choices.
	\end{proof}
	%%

	%%
	%%Section 3.2
	\subsection{Proof of Lemma \ref{lem: 3}}
	\label{Appendix: lemma 3 solution}
		\begin{proof}
	Partial derivative of the optimal  \(n_t\) with respect to \(\gamma_1\), \(\gamma_4\) and \(\gamma_6\) leads to 
	\begin{equation*}
		\frac{\partial n_t}{\partial \gamma_1} = \frac{\partial n_t}{\partial \gamma_4} = 
		\frac{\partial n_t}{\partial \gamma_6} = 
		-\frac{\gamma_2 + \gamma_5 - \gamma_3 }{\tau S^2}  < 0.
	\end{equation*}
	Assumption of Modified Quantity-Quality trade of condition lead to $\gamma_2 + \gamma_5 - \gamma_3 > 0$ and that follow result.
\end{proof}

	\subsection{Proof of Lemma \ref{lem: 4}}
	\label{Appendix: lemma 4 solution}
		\begin{proof}
	Partial derivative of the optimal  \(c_t\), \(p_t\) and \(n_t\) with respect to assigned weight of pension premium \(\gamma_7\) leads 
	to 
	\begin{eqnarray*}
		\frac{\partial c_t}{\partial \gamma_7} &=& -\frac{w_t \gamma_1}{S^2} < 0, \\
		\frac{\partial p_t}{\partial \gamma_7} &=& -\frac{w_t \gamma_4}{S^2} < 0,\\
		\frac{\partial n_t}{\partial \gamma_7} &=& -\frac{\gamma_2 + \gamma_5 - \gamma_3}{\tau S^2} < 0.
	\end{eqnarray*}
\end{proof}

	\subsection{Proof of Theorem \ref{lem: 5}}
	\label{Appendix: lemma 5 solution}
		\begin{proof}
	Partial derivative of the optimal  \(s_t\) with respect to \(\gamma_7\) leads to 
	\begin{equation*}
		\frac{\partial s_t}{\partial \gamma_7} = -\frac{w_ t \gamma_6 }{S^2} < 0.
	\end{equation*}
\end{proof}

	\subsection{Proof of Proposition \ref{lem: 6}}
	\label{Appendix: lemma 6 solution}
		\begin{proof}
	From the expression for the optimal number of children $
	n_t = \frac{\gamma_2 + \gamma_5 - \gamma_3}{\tau S}. 
	$
	Taking the partial derivative with respect to $\gamma_3$ leads to
	\begin{equation*}
		\frac{\partial n_t}{\partial \gamma_3} = -\frac{1}{\tau S} < 0. 
	\end{equation*}
\end{proof}

	\subsection{Proof of Proposition \ref{lem: 7}}
	\label{Appendix: lemma 7 solution}
		\begin{proof}
	The partial derivative of the optimal education expenditure per child ($e_t$) with respect to $\gamma_2$ is given by:
	\begin{equation*}
		\frac{\partial e_t}{\partial \gamma_2} = -\frac{w_t \gamma_3 \tau}{(\gamma_2 + \gamma_5 - \gamma_3)^2} < 0. 
	\end{equation*}
\end{proof}

	\subsection{Proof of Theorem \ref{lem: 8}}
	\label{Appendix: lemma 8 solution}
		\begin{proof}
	Solution of maximization problem lead to following optimal solutions:
	\begin{align*}
		c_t = \frac{\gamma_1 w_t}{S},  \quad
		s_t = \frac{\gamma_6 w_t}{S}, \quad
		p_t = \frac{\gamma_4 w_t}{S}, \quad
		q_t = \frac{\gamma_7 w_t}{S}, \quad  
		n_t =  \frac{ \gamma_2 + \gamma_5 - \gamma_3}{\tau S}.
	\end{align*}
	Upon taking the partial derivative of the above optimal solutions with respect to \(\gamma_5\), we have:
	\[
	\frac{\partial c_t}{\partial \gamma_5} = -\frac{w_t \gamma_1}{S^2} < 0, \qquad
	\frac{\partial s_t}{\partial \gamma_5} = -\frac{w_t \gamma_6}{S^2}  < 0,\]
	\[
	\frac{\partial p_t}{\partial \gamma_5} = -\frac{w_t \gamma_4}{S^2} < 0, \qquad
	\frac{\partial q_t}{\partial \gamma_5} = -\frac{w_t \gamma_7}{S^2} < 0,
	\]
	\[\frac{\partial n_t}{\partial \gamma_5} = \frac{ S-(\gamma_2+\gamma_5)}{\tau S^2} > 0.\]
	\end{proof}
\end{appendices}
%%
%%Apprndix ends here!!

%%Reference starts here!
%References 

\end{document}